\documentclass{article}

\usepackage{graphicx} 
\usepackage{dsfont}
\usepackage{stmaryrd}

\usepackage{graphicx}
\usepackage{subfigure}

\usepackage{multirow}
\usepackage{array}
\usepackage{amscd}

\usepackage{color}
\usepackage[standard]{ntheorem}

\newcommand{\X}{\mathcal{X}}

\usepackage{amsmath,amssymb}                      
\begin{document}

\title{Lyapunov Exponent Evaluation of the CBC Mode of Operation}

\author{Abdessalem Abidi, Christophe Guyeux, Jacques Demerjian,\\ Belagacem Bouall\`egue, and Mohsen Machhout}

\maketitle

\begin{abstract}
The Cipher Block Chaining (CBC) mode of encryption was invented in 1976, and it is currently one of the most commonly used mode. In our previous research works, we have proven that the CBC mode of operation exhibits,  under some conditions, a chaotic behavior. The dynamics of this mode has been deeply investigated later, both qualitatively and quantitatively, using the rigorous mathematical topology field of research. In this article, which is an extension of our previous work, we intend to compute a new important quantitative property concerning our chaotic CBC mode of operation, which is the Lyapunov exponent.
\end{abstract}

\section{Introduction}
Blocks ciphers, like Data Encryption Standard (DES) or 
Advanced Encryption Standard (AES), have a very simple 
principle: they do not treat the original text bit by bit 
but they manipulate blocks of text. More precisely, the 
plaintext is broken into blocks of $\mathsf{N}$ bits. For 
each one, the encryption algorithm is applied to obtain an 
encrypted block that has the same size. Then, we put 
together all of these blocks, which are separately 
encrypted, to obtain the full encrypted message. For 
decryption, we proceed in the same way, but now starting 
from the ciphertext, in order to obtain the original one 
employing the decryption algorithm in place of the 
encryption function. So it is not sufficient to put anyhow 
a block cipher algorithm in a program. We can, instead, use 
these algorithms in various ways according to their specific 
needs. These ways are called block cipher modes of 
operation. Indeed, there are several modes and each one of 
them differs from others by its own characteristics, in 
addition to its specific security properties. In this 
article, we are only interested in the Cipher Block 
Chaining mode and we will quantify its chaotic behavior 
thanks to the Lyapunov exponent. To do so, we first show that such mode of operations can be considered as dynamical systems.

Indeed, some dynamical systems are very sensitive to small changes in their initial
condition. Both constants of sensitivity to initial conditions and
of expansivity illustrate that~\cite{gb11:bc,guyeux10}. However, these
variations can quickly take enormous proportions, grow exponentially, and none of these constants can measure such a behavior. Alexander Lyapunov has examined this phenomenon and 
introduced an exponent that measures the rate at which these small variations can grow.

\begin{definition}
\label{def:lyapunov}
Let $f: \mathds{R} \longrightarrow \mathds{R}$.
The \emph{Lyapunov exponent} of the system defined by $x^0 \in \mathds{R}$ and
$x^{n+1} = f(x^n)$
 is:
$$\displaystyle{\lambda(x_0)=\lim_{n \to +\infty} \dfrac{1}{n} \sum_{i=1}^n \ln \left| ~f'\left(x^{i-1}\right)\right|}.$$
\end{definition}

Consider a dynamical system with an infinitesimal error on the initial condition $x_0$. 
When the Lyapunov exponent is positive, this
error will increase (situation of chaos), whereas it will decrease if $\lambda(x_0)\leqslant 0$.

\begin{example}
The Lyapunov exponent of the logistic map $x^0 \in [0,1]$, $x^{n+1}=\mu x^n (1-x^n)$~\cite{arroyo2008inadequacy} becomes positive for
$\mu>3,54$, but it is always smaller than 1. The tent map~\cite{wang2009block,Guyeux2012} 
and the doubling map of the circle~\cite{richeson2008chain}, two other well-known chaotic dynamical systems, have a Lyapunov exponent equal to $\ln(2)$.
\end{example}

Sometimes, instead of trying to prove  directly the 
properties on the system itself, it is preferable to 
reduce the initial problem to another whose characteristics 
are known or seem to be accessible. Such a reduction tool 
is called, in the mathematical theory of chaos, the 
semi-conjugacy.

\begin{definition}
\label{def:lyapunov}
The discrete dynamical system $(\mathcal{X},f)$ is  \emph{topologically
semi-conjugate} to the system $(\mathcal{Y},g)$ if it exists a 
function $\varphi : \mathcal{X} \longrightarrow
\mathcal{Y}$, both continuous and onto, such that: $$\varphi \circ f = g \circ \varphi,$$ \noindent
that is, which makes commutative the following diagram~\cite{Formenti1998}.

\begin{equation*}
\begin{CD}
\mathcal{X} @>f>>\mathcal{X}\\
    @V{\varphi}VV                    @VV{\varphi}V\\
\mathcal{Y} @>>g> \mathcal{Y}
\end{CD}
\end{equation*}

In this case, the system $(\mathcal{Y},g)$ is called a \emph{factor} of the system
$(\mathcal{X},f)$.
\end{definition}

Various dynamical behaviors are inherited by systems factors~\cite{Formenti1998}.
They are summarized in the following proposition:

\begin{proposition}
\label{prop:factor-system-chaos-preservation}
Let $(\mathcal{Y},g)$ a factor of the system $(\mathcal{X},f)$. Then:
\begin{enumerate}
 \item for all $j \leqslant k$, $p \in Per_k(f) \Longrightarrow \varphi(p) \in Per_j(g)$, 
 where $Per_n(h)$ stands for the set of points of period $n$ for 
 the iteration function $h$.
\item $(\mathcal{X},f)$ regular $\Longrightarrow$ $(\mathcal{Y},g)$ regular,
\item $(\mathcal{X},f)$ transitive $\Longrightarrow$ $(\mathcal{Y},g)$ transitive.
\end{enumerate}
So if $(\mathcal{X},f)$ is chaotic as defined by Devaney, then
$(\mathcal{Y},g)$ is chaotic too.
\end{proposition}

Having these materials in mind, it is now possible to measure the Lyapunov exponent of some CBC mode of operations. Do do so, we will follow the canvas described hereafter.
In Section
\ref{sec:basic-recalls}, some basic reminders are given.
The semiconjugacy allowing the 
exponent evaluation is described in Section~\ref{sec:topological-semiconjugacy}. In the next one, the
consequences of such a semi-conjugacy are outlined, and the exponent is computed. This article ends by a conclusion section where our contribution is summarized and intended future work is outlined.

\section{Basic Recalls}
\label{sec:basic-recalls}

\subsection{The Cipher Block Chaining (CBC) mode}
 \label{sec:CBC properties}

The CBC block cipher mode of operation presents a very popular way of encryption that is used in numerous applications, despite the fact that encryption in this mode can be performed only using one thread. Cipher block chaining is a block cipher mode that provides confidentiality but not message integrity in cryptography. 
The operating principle of this mode is to add XOR each subsequent plain-text block to a cipher-text one that was previously received, see Figure~\ref{fig:CBC}. 
Each subsequent cipher-text block depends on the previous one. Finally, the first plain-text block is added XOR to a random Initialization Vector (commonly referred to as IV). This vector has the same size as all plain-text blocks. 

To decrypt cipher-text blocks, one should add XOR output data from decryption algorithm to previous cipher-text blocks. The receiver knows all cipher-text blocks just after obtaining encoded the message, thus he can decrypt the message using many threads simultaneously.
If one bit of a plain-text message is damaged (for instance, because of some earlier transmission error), all subsequent cipher-text blocks will be damaged and it will be never possible to decrypt the cipher-text received from this plain-text. As opposed to that, if one cipher-text bit is damaged, only two received plain-text blocks will be damaged.

\begin{figure}
    \centering
 \subfigure[CBC encryption mode]{\label{fig:CBCenc}
        \includegraphics[scale=0.3]{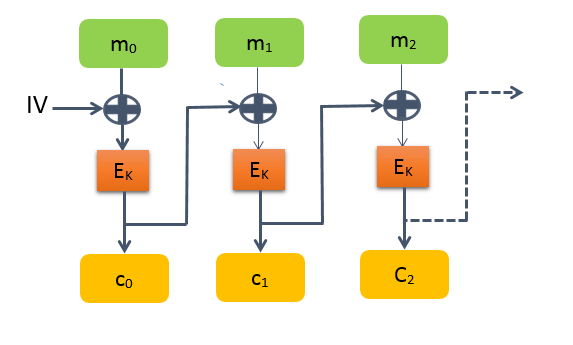}}     \subfigure[CBC decryption mode]{\includegraphics[scale=0.3]{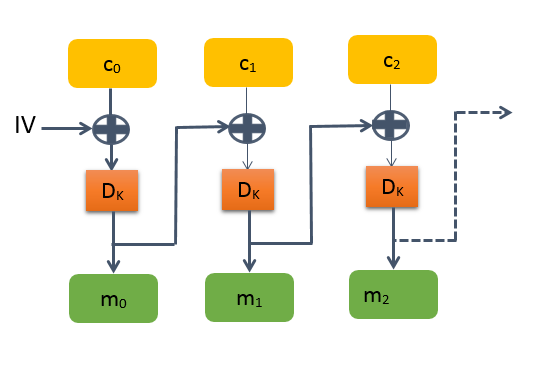}}
    \caption{CBC mode of operation}
     \label{fig:CBC}
\end{figure}



Finally, note that a message that is to be encrypted using the CBC mode, should be extended until being as long as a multiple of a single block length.

\subsection{Modeling the CBC mode as a dynamical system}
Our modeling follows a same canvas than what has be done for hash functions~\cite{bg10:ij,gb11:bc} or pseudo-random number generation~\cite{bfgw11:ij}.
Let us consider the CBC mode of operation with a keyed encryption function $\varepsilon_k:\mathds{B}^\mathsf{N} \rightarrow \mathds{B}^\mathsf{N} $ depending on a secret key $k$, where $\mathsf{N}$ is the size for the block cipher, and $\mathcal{D}_k:\mathds{B}^\mathsf{N} \rightarrow \mathds{B}^\mathsf{N} $ is the associated  decryption function, which is such that $\forall k, \varepsilon_k \circ \mathcal{D}_k$ is the identity function. We define  
the Cartesian product $\mathcal{X}=\mathds{B}^\mathsf{N}\times\mathcal{S}_\mathsf{N}$, where:
\begin{itemize}
\item $\mathds{B} = \{0,1\}$ is the set of Boolean values,
\item $\mathcal{S}_\mathsf{N} = \llbracket 0, 2^\mathsf{N}-1\rrbracket^\mathds{N}$, the set of infinite sequences of natural integers bounded by $2^\mathsf{N}-1$, or the set of infinite $\mathsf{N}$-bits block messages,
\end{itemize}
in such a way that $\mathcal{X}_\mathsf{N}$ is constituted by couples of internal states of the mode of operation together with sequences of block messages.
Let us consider the initial function:
$$\begin{array}{cccc}
 i:& \mathcal{S}_\mathsf{N} & \longrightarrow & \llbracket 0, 2^\mathsf{N}-1 \rrbracket \\
 & (m^i)_{i \in \mathds{N}} & \longmapsto & m^0
\end{array}$$
that returns the first block of a (infinite) message, and the shift function:
$$\begin{array}{cccc}
 \sigma:& \mathcal{S}_\mathsf{N} & \longrightarrow & \mathcal{S}_\mathsf{N} \\
 & (m^0, m^1, m^2, ...) & \longmapsto & (m^1, m^2, m^3, ...)
\end{array}$$
which removes the first block of a message. Let $m_j$ be the $j$-th bit of integer, or block message, $m\in \llbracket 0, 2^\mathsf{N}-1 \rrbracket$, expressed in the binary numeral system, and when counting from the left. We define:
$$\begin{array}{cccc}
F_f:& \mathds{B}^\mathsf{N}\times \llbracket 0, 2^\mathsf{N}-1 \rrbracket & \longrightarrow & \mathds{B}^\mathsf{N}\\
 & (x,m) & \longmapsto & \left(x_j m_j + f(x)_j \overline{m_j}\right)_{j=1..\mathsf{N}} 
\end{array}$$
This function returns the inputted binary vector $x$, whose $m_j$-th components $x_{m_j}$ have been replaced by $f(x)_{m_j}$, for all $j=1..\mathsf{N}$ such that $m_j=0$. In case where $f$ is the vectorial negation, this function will correspond to one XOR between the clair text and the previous encrypted state.  

Denote by $f_0$ the vectorial negation. So the CBC mode of operation can be rewritten n a condensed way, as follows.
\begin{equation}
\label{eq:sysdyn}
\left\{
\begin{array}{ll}
X^0 = & (IV,m)\\
X^{n+1} = & \left(\mathcal{E}_k \circ F_{f_0} \left( i(X_1^n), X_2^n\right), \sigma (X_1^n)\right)
\end{array}
\right.
\end{equation}
For any given $g:\llbracket 0, 2^\mathsf{N}-1\rrbracket \times \mathds{B}^\mathsf{N} \longrightarrow \mathds{B}^\mathsf{N}$, we denote $G_g(X) = \left(g(i(X_1),X_2);\sigma (X_1)\right)$ (when $g = \mathcal{E}_k\circ F_{f_0}$, we obtain one cypher block of the CBC, as depicted in Figure~\ref{fig:CBC}). So the recurrent relation of Eq.\eqref{eq:sysdyn} can be rewritten in a condensed way, as follows.
\begin{equation}
X^{n+1} = G_{\mathcal{E}_k\circ F_{f_0}} \left(X^n\right) .
\end{equation}
With such a rewriting, one iterate of the discrete dynamical system above corresponds exactly to one cypher block in the CBC mode of operation. Note that the second component of this system is a subshift of finite type, which is related to the symbolic dynamical systems known for their relation with chaos~\cite{lind1995introduction}.

We then have defined a distance on $\mathcal{X}_\mathsf{N}$ as follows: $d((x,m);(\check{x},\check{m})) = d_e(x,\check{x})+d_m(m,\check{m})$, where~\cite{abidi2016proving}:
$$\left\{\begin{array}{ll}
d_e(x,\check{x})  & = \sum_{k=1}^\mathsf{N} \delta (x_k,\check{x}_k)  \\
&\\
d_m(m,\check{m}) & = \displaystyle{\dfrac{9}{\mathsf{N}} \sum_{k=1}^\infty \dfrac{\sum_{i=1}^\mathsf{N} \left|m_i - \check{m}_i\right|}{10^k}} .
\end{array}\right.$$
in which $\delta(x,y)=1$ if $x=y$, else it is 0. Using this modeling, we have been able to prove that~\cite{abidi2016proving},
\begin{theorem}
\label{th:chaos selon Devaney}
The CBC mode of operation $G_{\mathcal{E}_k\circ F_{f_0}}$ is chaotic, as defined by Devaney~\cite{devaney}, on the topological space $(\mathcal{X},d)$. This means that $G_{\mathcal{E}_k\circ F_{f_0}}$ has  on $(\mathcal{X}, d)$ the properties of:
\begin{itemize}
  \item \emph{regularity}: its set of periodic
points is dense in $\mathcal{X}_\mathsf{N}$ (for any point $x$ in $\mathcal{X}_\mathsf{N}$,
any neighborhood of $x$ contains at least one periodic point).
\item \emph{topologically transitivity}: for any pair of open sets
$U,V \subset \mathcal{X}_\mathsf{N}$, there exists an integer $k>0$ such that $G_{\mathcal{E}_k\circ F_{f_0}}^k(U) \cap V \neq \varnothing$.
\item \emph{sensitive dependence on initial conditions}: there exists $\delta >0$ such that, for any $x\in \mathcal{X}_\mathsf{N}$ and any neighborhood $V$ of $x$, there exist $y\in V$ and $n > 0$ such that
$$d\left(G_{\mathcal{E}_k\circ F_{f_0}}^{n}(x), G_{\mathcal{E}_k\circ F_{f_0}}^{n}(y)\right) >\delta .$$
\end{itemize}
\end{theorem}
This result has been extended in~\cite{abidi2016quantitative}, in which both expansivity and sensibility of symmetric cyphers have been regarded in the case of the CBC mode of operation. However, all these results of qualitative and quantitative disorder have been stated on an exotic phase space $\mathcal{X}_\mathsf{N}$, equipped with a distance $d$ very different from the usual Euclidian one.
 Our objective is now to translate them in a more usual situation, namely the real line equipped with its usual order topology. To do so, a topological semi-conjugacy must be introduced. Such a formulation will make it possible to evaluate the Lyapunov exponent of the CBC mode,
as the latter will be described by a differentiable function on $\mathds{R}$.

\begin{figure}[t]
\begin{center}
  \subfigure[Function $x \to dist(x;1.5) $ on the interval
$(0;4)$.]{\includegraphics[scale=.27]{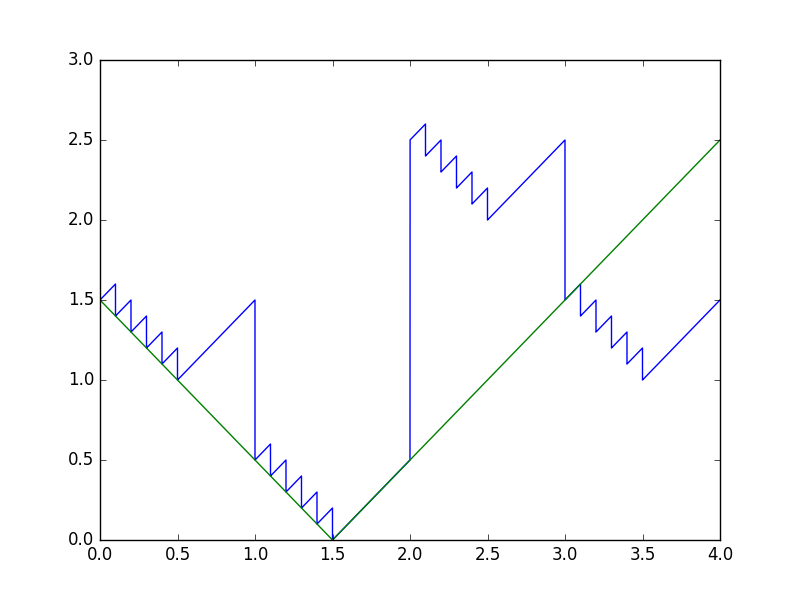}}\quad
  \subfigure[Function $x \to dist(x;1.9) $ on the interval
$(0;4)$.]{\includegraphics[scale=.27]{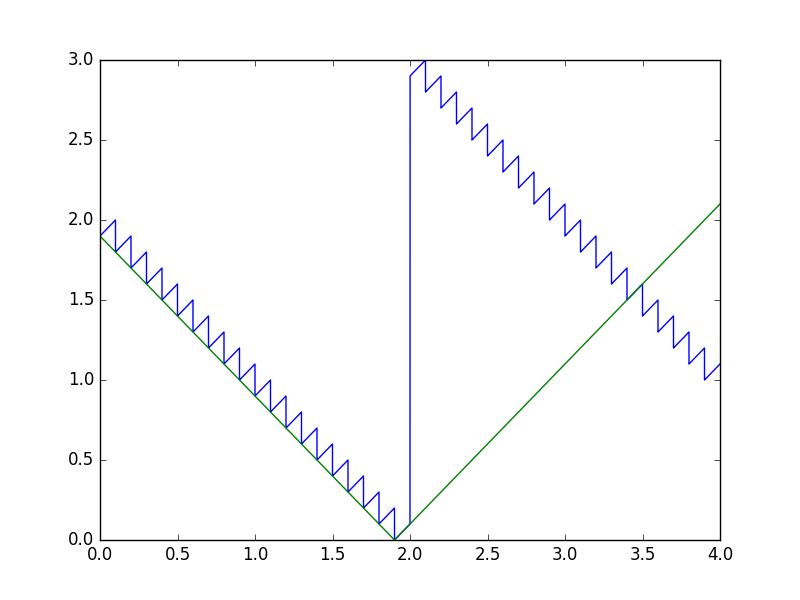}}
\end{center}
\caption{Comparison between $D$ (in blue) and the Euclidean distance (in
green).}
\label{fig:comparaison de distances}
\end{figure}



\section{A Topological Semi-conjugacy}
\label{sec:topological-semiconjugacy}
\subsection{The phase space is an interval of the real line}

\subsubsection{Toward a topological semi-conjugacy}

We show, by using a topological
semi-conjugacy, that CBC mode can be
described on a real interval. 
In what follows and for easy understanding, 
we will assume that $\mathsf{N} = 10$.
However, an  equivalent formulation of the following can be easily obtained  by replacing the base $10$ by any base $\mathsf{N}$.

\begin{definition}
The function $\varphi: \mathcal{S}_{10} \times\mathds{B}^{10}
\rightarrow \big[
0, 2^{10} \big[$ is defined by:
\begin{equation*}
 \begin{array}{cccl}
\varphi: & \mathcal{X}_{10} = \mathcal{S}_{10} \times\mathds{B}^{10}&
\longrightarrow & \big[ 0, 2^{10} \big[ \\
 & \left((S_0, S_1, \hdots ); (E_0, \hdots, E_9)\right) &
\longmapsto &
\varphi \left((S,E)\right)
\end{array}
\end{equation*}
where $(S,E) = \left((S_0, S_1, \hdots ); (E_0, \hdots, E_9)\right)$, and
$\varphi\left((S,E)\right)$ is the real number:
\begin{itemize}
\item whose integral part $e$ is $\displaystyle{\sum_{k=0}^9 2^{9-k}
E_k}$, that
is, the binary digits of $e$ are $E_0 ~ E_1 ~ \hdots ~ E_9$.
\item whose decimal part $s$ is equal to $s = 0,S_0~ S_1~ S_2~ \hdots =
\sum_{k=1}^{+\infty} 10^{-k} S^{k-1}.$ 
\end{itemize}
\end{definition}

\noindent $\varphi$ realizes the association between a point of $\mathcal{X}_{10}$
and a
real number into $\big[ 0, 2^{10} \big[$. We must now translate the CBC process $G_{\mathcal{E}_k\circ F_{f_0}}$ on this real interval. To do so, two intermediate
functions
over $\big[ 0, 2^{10} \big[$ must be introduced:

\begin{definition}
\label{def:e et s}
Let $x \in \big[ 0, 2^{10} \big[$ and:
\begin{itemize}
\item $e_0, \hdots, e_9$ the binary digits of the integral part of $x$:
$\displaystyle{\lfloor x \rfloor = \sum_{k=0}^{9} 2^{9-k} e_k}$.
\item $(s_k)_{k\in \mathds{N}}$ the digits of $x$, where the chosen
decimal
decomposition of $x$ is the one that does not have an infinite number of
9: 
$\displaystyle{x = \lfloor x \rfloor + \sum_{k=0}^{+\infty} s_k
10^{-k-1}}$.
\end{itemize}
$e$ and $s$ are thus defined as follows:
\begin{equation*}
\begin{array}{cccl}
e: & \big[ 0, 2^{10} \big[ & \longrightarrow & \mathds{B}^{10} \\
 & x & \longmapsto & (e_0, \hdots, e_9)
\end{array}
\end{equation*}
and
\begin{equation*}
 \begin{array}{cccc}
s: & \big[ 0, 2^{10} \big[ & \longrightarrow & \llbracket 0, 9
\rrbracket^{\mathds{N}} \\
 & x & \longmapsto & (s_k)_{k \in \mathds{N}}
\end{array}
\end{equation*}
\end{definition}

We are now able to define the function $g$, whose goal is to translate
the
CBC mode $G_{\mathcal{E}_k\circ F_{f_0}}$ on an interval of $\mathds{R}$.

\begin{definition}\label{def:function-g-on-R}
$g:\big[ 0, 2^{10} \big[ \longrightarrow \big[ 0, 2^{10} \big[$ is
defined by:
\begin{equation*}
\begin{array}{cccc}
g: & \big[ 0, 2^{10} \big[ & \longrightarrow & \big[ 0, 2^{10} \big[ \\
 & x & \longmapsto & g(x)
\end{array}
\end{equation*}
where g(x) is the real number of $\big[ 0, 2^{10} \big[$ defined bellow:
\begin{itemize}
\item its integral part is the number, encrypted by $\varepsilon_\mathsf{k}$, whose binary decomposition equal to $e_0',
\hdots,
e_9'$, with:
 \begin{equation*}
e_i' = \left\{
\begin{array}{ll}
e(x)_i & \textrm{ if } m_i^0 = 0\\
e(x)_i + 1 \textrm{ (mod 2)} & \textrm{ if } m_i^0 = 1\\
\end{array}
\right.
\end{equation*}
\item whose decimal part is $m_0^1, \hdots, m_9^1$, 
$m_0^2, \hdots, m_9^2,m_0^3, \hdots, m_9^3, \hdots$
\end{itemize}
\end{definition}

In other words, if $x = \displaystyle{\sum_{k=0}^{9} 2^{9-k} e_k + 
\sum_{k=0}^{+\infty} s^{k} ~10^{-k-1}}$, then:
\begin{equation*}
g(x) =
\displaystyle{\sum_{k=0}^{9} 2^{9-k} \varepsilon_\mathsf{k} (e_k + \delta(k,s_0) \textrm{ (mod
2)}) + 
\sum_{k=0}^{+\infty} s^{k+1} 10^{-k-1}}. 
\end{equation*}

\subsubsection{Defining a metric on $\big[ 0, 2^{10} \big[$}

Numerous metrics can be defined on the set $\big[ 0, 2^{10} \big[$, the
most
usual one being the Euclidean distance
$\Delta(x,y) = \sqrt{y^2-x^2}$.
This Euclidean distance does not reproduce exactly the notion of
proximity
induced by our first distance $d$ on $\X$. Indeed $d$ is finer than
$\Delta$.
This is the reason why we have to introduce the following metric:

\begin{definition}
Let $x,y \in \big[ 0, 2^{10} \big[$.
$D$ denotes the function from $\big[ 0, 2^{10} \big[^2$ to $\mathds{R}^
+$
defined by: $D(x,y) = D_e\left(e(x),e(y)\right) + D_s
\left(s(x),s(y)\right)$,
where:
\begin{center}
$\displaystyle{D_e(E,\check{E}) = \sum_{k=0}^\mathsf{9} \delta (E_k,
\check{E}_k)}$, ~~and~ $\displaystyle{D_s(S,\check{S}) = \sum_{k = 1}^
\infty
\dfrac{|s_k-\check{S}^k|}{10^k}}$.
\end{center}
\end{definition}

\begin{proposition}
$D$ is a distance on $\big[ 0, 2^{10} \big[$.
\end{proposition}

\begin{proof}
The three axioms defining a distance must be checked.
\begin{itemize}
\item $D \geqslant 0$, because everything is positive in its definition.
If
$D(x,y)=0$, then $D_e(x,y)=0$, so the integral parts of $x$ and $y$ are
equal
(they have the same binary decomposition). Additionally, $D_s(x,y) = 0$,
then
$\forall k \in \mathds{N}^*, s(x)^k = s(y)^k$. In other words, $x$ and
$y$ have
the same $k-$th decimal digit, $\forall k \in \mathds{N}^*$. And so $x=y
$.
\item $D(x,y)=D(y,x)$.
\item Finally, the triangular inequality is obtained due to the fact
that both
$\delta$ and $|x-y|$ satisfy it.
\end{itemize}
\end{proof}

The convergence of sequences according to $D$ is not the same than the
usual convergence related to the Euclidean metric. For 
instance, if $x^n \to x$ according to $D$, then necessarily 
the integral part of each $x^n$ is equal to the integral 
part of $x$ (at least after a given threshold), and the decimal part of $x^n$ corresponds to the one of $x$ ``as far as required''.
To illustrate this fact, a comparison between $D$ and the Euclidean
distance is
given in Figure~\ref{fig:comparaison de distances}. These illustrations show that $D$ is richer and more refined than the Euclidean distance, and thus is more precise.

\subsubsection{The semi-conjugacy}

It is now possible to define a topological semi-conjugacy between
$\mathcal{X}_\mathsf{N}$ 
and an interval of $\mathds{R}$ which makes possible to 
translate the action of the CBC encryption on a message in 
the form of a recurrent sequence on the interval $\big[ 0, 2^{10} \big[$.

\begin{theorem}
CBC mode on the phase space $\mathcal{X}_\mathsf{N}$ are simple
iterations on
$\mathds{R}$, which is illustrated by the semi-conjugacy of the diagram
below:
\begin{equation*}
\begin{CD}
\left(~\mathcal{S}_{10} \times\mathds{B}^{10}, d~\right) @>G_{\mathcal{E}_k\circ F_{f_0}}>>
\left(~\mathcal{S}_{10} \times\mathds{B}^{10}, d~\right)\\
    @V{\varphi}VV                    @VV{\varphi}V\\
\left( ~\big[ 0, 2^{10} \big[, D~\right)  @>>g> \left(~\big[ 0, 2^{10}
\big[,
D~\right)
\end{CD}
\end{equation*}
\end{theorem}

\begin{proof}
$\varphi$ has been constructed in order to be continuous and onto.
\end{proof}

In other words, $\mathcal{X}_\mathsf{N}$ is approximately equal to $\big[ 0, 2^
\mathsf{N}
\big[$.

\subsubsection{Comparing the metrics of $\big[ 0, 2^\mathsf{N} \big[$}

The two propositions below allow us to compare our two distances on $\big[ 0, 2^\mathsf{N} \big[$:

\begin{proposition}
The identity function Id: $\left(~\big[ 0, 2^\mathsf{N} \big[,\Delta~\right) \to \left(~\big[ 0, 2^\mathsf{N} \big[, D~\right)$ is not continuous. 
\end{proposition}

\begin{proof}
The sequence $x^n = 1,999\hdots 999$ constituted by $n$ 9's as digits, is such that:
\begin{itemize}
\item $\Delta (x^n,2) \to 0.$
\item But $D(x^n,2) \geqslant 1$, so $D(x^n,2)$ does not converge to 0.
\end{itemize}

The sequential characterization of the continuity allows us to conclude the proposition.
\end{proof}

A contrario:

\begin{proposition}
Id: $\left(~\big[ 0, 2^\mathsf{N} \big[,D~\right) \to \left(~\big[ 0, 2^\mathsf{N} \big[, \Delta ~\right)$ is continuous. 
\end{proposition}

\begin{proof}
On the one hand, if $D(x^n,x) \to 0$, then $D_e(x^n,x) = 0$ at least after a given rank, because $D_e$ produces only integers. So, after a given rank, the whole integral parts of $x^n$ are equal to the one of $x$. 

On the other hand, $D_s(x^n, x) \to 0$, so $\forall k \in \mathds{N}^*, \exists N_k \in \mathds{N}, n \geqslant N_k \Rightarrow D_s(x^n,x) \leqslant 10^{-k}$. Which means that for all $k$, it exists a rank  $N_k$ after which all the $x^n$'s have the same $k$ first digits, which are the ones of $x$. We can deduce from all these aspects that $\Delta(x^n,x) \to 0$, which leads to the claimed result.
\end{proof}

We can conclude from the previous propositions that the introduced metric is more precise than the Euclidean distance. In other words:

\begin{proposition}
The distance $D$ is finer than the Euclidean distance $\Delta$.
\end{proposition}

This proposition can be reformulated as follows:
\begin{itemize}
\item The topology generated by $\Delta$ is inside the one generated by $D$.
\item $D$ has more open sets than $\Delta$.
\item Figuratively, $D$ allows a better observation, leading to more details than $\Delta$.
\item Finally, it is harder to converge with the topology $\tau_D$ generated by $D$, than with the one generated by $\Delta$, and denoted $\tau_\Delta$.
\end{itemize}

\subsubsection{Impact of the topology}
To alleviate notations, let us denote by $\mathcal{X}_\tau$ the topological space $\left(\mathcal{X},\tau\right)$, and by $\mathcal{V}_\tau (x)$ the set of all neighborhoods of $x$ when considering the $\tau$ topology. When there is no ambiguity, we will simply use the notation $\mathcal{V} (x)$.

\begin{theorem}
\label{Th:chaos et finesse}
Let $\mathcal{X}_\mathsf{N}$ be a set, and $\tau, \tau'$ two topologies on $\mathcal{X}_\mathsf{N}$ such that $\tau'$ is finer than $\tau$. Let $f:\mathcal{X} \to \mathcal{X}_\mathsf{N}$ be a function continuous for both $\tau$ and $\tau'$.

If $(\mathcal{X}_{\tau'},f)$ is chaotic according to Devaney, then $(\mathcal{X}_\tau,f)$ is chaotic too.
\end{theorem}

\begin{proof}
Let us firstly introduce the transitivity of $(\mathcal{X}_\tau,f)$.

Let $\omega_1, \omega_2$ be two open sets of $\tau$. Then $\omega_1, \omega_2 \in \tau'$, as $\tau'$ is finer than $\tau$. But $f$ is $\tau'-$transitive, so we can deduce that $\exists n \in \mathds{N}, \omega_1 \cap f^{(n)}(\omega_2) = \varnothing$. As a consequence, $f$ is $\tau-$transitive.

Let us now establish the regularity of $(\mathcal{X}_\tau,f)$, \emph{i.e.}, for all $x \in \mathcal{X}_\mathsf{N}$, and for all $\tau-$neighborhood $V$ of $x$, a periodic point for $f$ can be found in $V$.

Let $x \in \mathcal{X}_\mathsf{N}$ and $V \in \mathcal{V}_\tau (x)$ a $\tau-$neighborhood of $x$. By definition of the neighborhood notion, $\exists \omega \in \tau, x \in \omega \subset V$.

But $\tau \subset \tau'$, so $\omega \in \tau'$, and as a consequence, $V \in \mathcal{V}_{\tau'} (x)$. As $(\mathcal{X}_{\tau'},f)$ is regular, it exists a periodic point for $f$ in $V$, and the regularity of $(\mathcal{X}_\tau,f)$ is proven.
\end{proof}

\subsection{CBC mode described as a real function}
We will now show that the $g$  function is a piecewise linear one:
it is
linear on each interval having the form $\left[ \dfrac{n}{10},
\dfrac{n+1}{10}\right[$, $n \in \llbracket 0;2^{10}\times 10 \rrbracket$
and its
slope is equal to 10. 

\begin{proposition}
\label{Prop:derivabilite des ICs}
CBC mode $g$ defined on $\mathds{R}$ have derivatives of all
orders on
$\big[ 0, 2^{10} \big[$, except on the 10241 points in $I$ defined by
$\left\{
\dfrac{n}{10} ~\big/~ n \in \llbracket 0;2^{10}\times 10\rrbracket
\right\}$.

Furthermore, on each interval of the form $\left[ \dfrac{n}{10},
\dfrac{n+1}{10}\right[$, with $n \in \llbracket 0;2^{10}\times 10
\rrbracket$,
$g$ is a linear function, having a slope equal to 10: $\forall x \notin
I,
g'(x)=10$.
\end{proposition}

\begin{proof}
Let $I_n = \left[ \dfrac{n}{10}, \dfrac{n+1}{10}\right[$, with $n \in
\llbracket
0;2^{10}\times 10 \rrbracket$. All the points of $I_n$ have the same
integral
part $e$ and the same decimal part $s_0$: on the set $I_n$,  functions
$e(x)$
and $x \mapsto s(x)^0$ of Definition \ref{def:e et s} only depend on $n
$. So all
the images $g(x)$ of these points $x$:
\begin{itemize}
\item Have the same integral part, which is $\varepsilon_\mathsf{k}(e)$, except probably the bit
number
$s_0$. In other words, this integer has approximately the same binary
decomposition than $\varepsilon_\mathsf{k}(e)$, the sole exception being the digit $s_0$ (this
number is
then either $\varepsilon_\mathsf{k}(e+2^{10-s_0})$ or $\varepsilon_\mathsf{k}(e-2^{10-s_0})$, depending on the parity of
$s_0$,
\emph{i.e.}, it is equal to $\varepsilon_\mathsf{k}(e+(-1)^{s_0}\times 2^{10-s_0})$).
\item A shift to the left has been applied to the decimal part $y$,
losing by
doing so the common first digit $s_0$. In other words, $y$ has been
mapped into
$10\times y - s_0$.
\end{itemize}
To sum up, the action of $g$ on the points of $I$ is as follows: first,
make a
multiplication by 10, and second, add the same constant to each term,
which is
$\dfrac{1}{10}\left(\varepsilon_\mathsf{k}(e+(-1)^{s_0}\times 2^{10-s_0})\right)-s_0$.
\end{proof}

\begin{remark}
CBC mode is then an element of the large family of
functions that are both chaotic and piecewise linear (like the tent map~\cite{wang2009block,Guyeux2012}).
\end{remark}
We are now able to evaluate the Lyapunov exponent of our chaotic CBC mode, which is now described by the iterations on $\mathds{R}$
of the $g$ function introduced in Definition~\ref{def:function-g-on-R}.

\section{Disorder generated by CBC formulated on $\mathds{R}$}
\label{chpt:Chaos des itérations chaotiques sur R}

\subsection{Devaney's chaos on the real line}

We have established in~\cite{abidi2016proving} that the CBC mode of operation 
$\left(G_{\mathcal{E}_k\circ F_{f_0}}, \mathcal{X}_d\right)$ satisfies the Devaney's definition of chaos. From the semi-conjugacy, we can deduce that it is the case too for the 
mode of operation on $\mathds{R}$ with the order topology, as:
\begin{itemize}
\item $\left(G_{\mathcal{E}_k\circ F_{f_0}}, \mathcal{X}_d\right)$ and $\left(g, \big[ 0, 2^{10} \big[_D\right)$ are semi-conjugated by $\varphi$,
\item $\varphi\left(g, \big[ 0, 2^{10} \big[_D\right)$ is a chaotic system according to Devaney, because the semi-conjugacy preserves such a character~\cite{Formenti1998}.
\item But the topology generated by $D$ is finer than the one generated by the euclidean distance $\Delta$ -- which is the order topology~\cite{Guyeux2012}.
\item According to Theorem~\ref{Th:chaos et finesse}, we can deduce that the CBC mode of operation $g$ is chaotic, as defined by Devaney, for the usual order topology on $\mathds{R}$.
\end{itemize}

We can formulate this result as follows.

\begin{theorem}
\label{th:IC et topologie de l'ordre}
The CBC mode of operation $g$ on $\mathds{R}$ satisfies the Devaney's chaos property, when $\mathds{R}$ is equipped with its usual topology (the order one).
\end{theorem}

Indeed this result is weaker than Theorem~\ref{th:chaos selon Devaney}, that established the chaos of iterates on a finer topology. This can be explained in the following figurative manner. By using tools that are usual in the discrete dynamical system field, we can only observe disorder in the iterations of the CBC mode of operation (Theorem~\ref{th:IC et topologie de l'ordre}). And even if we considered an higher resolution, and more powerful tools than the ones that are commonly used, we still fail in finding order in such a chaos (Theorem~\ref{th:chaos selon Devaney}).

Result of Theorem~\ref{th:IC et topologie de l'ordre} is still precious. Indeed, we have started to formulate the mode of operation on a set different from the one commonly considered ($\mathcal{X}_\mathsf{N}$ instead of $\mathds{R}$), to be as close as possible to the computer machine (dealing with bounded integer), and so to prevent from losing disorder properties when switching from theory to computer program. It is to be feared that this introduction of discrete iterations can only be paid by the obtention of disorders of lower quality. In other words, perhaps we moved from a situation of a good disorder lost when computed on finite state machines, to a disorder preserved but of poor quality. Theorem~\ref{th:IC et topologie de l'ordre} shows exactly the contrary of this claim.

\subsection{Evaluation of the Lyapunov Exponent}

Let $\mathcal{L} = \left\{ x^0 \in \big[ 0, 2^{10} \big[ ~ \big/ ~ \forall n
\in \mathds{N}, x^n \notin I \right\}$, where $I$ is the set of points in the
real interval where $g$ is not differentiable (as it is explained in Proposition
\ref{Prop:derivabilite des ICs}). We have the following result.

\begin{theorem}
Let us consider the CBC mode of operation with block size of $N$. Then, $\forall x^0 \in \mathcal{L}$, its Lyapunov exponent is equal to $\lambda(x^0) = \ln (N)$.
\end{theorem}

\begin{proof}
The function $g$ is piecewise linear, with a slop of 10, as $g'(x)=10$
where $g$ is differentiable. Then $\forall x \in
\mathcal{L}$, $\lambda (x) = \lim_{n \to +\infty} \dfrac{1}{n} \sum_{i=1}^n
 \ln \left| ~g'\left(x^{i-1}\right)\right|
=  \lim_{n \to +\infty} \dfrac{1}{n} \sum_{i=1}^n \ln \left|10\right| 
= \lim_{n\to +\infty} \dfrac{1}{n} n \ln \left|10\right| = \ln 10.$
\end{proof}

\begin{remark}
The set of initial vectors for which this exponent is not defined is countable.
This is indeed the initial conditions such that an iteration value will be a
number having the form $\dfrac{n}{10}$, with $n\in \mathds{N}$. We can reach such a
 real number only by starting iterations on a \emph{decimal number}, as this latter must have
a finite fractional part.
\end{remark}

\begin{remark}
For a system having $\mathsf{N}$
cells, we will find, mutatis, an infinite uncountable set of initial conditions $x^0 \in \left[0;2^\mathsf{N}\right[$ such that
$\lambda (x^0) = \ln (\mathsf{N})$.
\end{remark}

So, it is possible to make the Lyapunov exponent of our CBC mode as large as possible, depending on the size of the block message.

\section{Conclusion and Future work}
We have available now a new quantitative property concerning the CBC mode of operation: its Lyapunov exponent is equal to ln(N), where N is the size of the block message.
This exponent allows to quantify how the ignorance on the exact initial vector increases after several iterations of the mode of operation. It illustrates the disorder
generated by iterations of such a process, reinforcing its chaotic nature.

Using the semi-conjugacy described here, it will be possible in a
future work to compare the topological behavior of various modes of operation on
$\mathcal{X}_\mathsf{N}$ and on $\mathds{R}$. This semi-conjugacy can be used to investigate various interesting directions, as to have a new understanding of the modes of operations while considering them as iterations on the real line. Their dynamics can be better understood thanks to the use of mathematical analyzis tools. Finally, elements of comparison with usual iteration ways can be provided too, as we will consider the same iteration set, namely the real line.

\bibliographystyle{unsrt}
\bibliography{biblio}
\end{document}